\documentclass[a4paper,UKenglish,cleveref, autoref, thm-restate]{lipics-v2019}

\usepackage{xspace}
\usepackage{upgreek}
\usepackage{color, colortbl}
\usepackage{mathtools}
\usepackage[linesnumbered,ruled,vlined]{algorithm2e}
\usepackage{xspace}
\usepackage{multirow}
\usepackage{makecell}

\newcommand{\set}[1]{\left\{ #1 \right\}}

\newcommand{\spar}[1]{\left[ #1 \right]}

\newcommand{\integer}{\mathbb{Z}}
\renewcommand{\natural}{\mathbb{N}}

\DeclareMathOperator{\poly}{poly}

\DeclareMathOperator{\Win}{Win}

\def\eqdef{\mathrel{{:}{=}}}
\def\e#1{\emph{#1}}
\def\T{\mathbf{T}}
\def\Q{\mathbf{Q}}
\def\M{\mathbf{M}}

\def\np{\text{$\mathsf{NP}$}\xspace}
\def\sharpp{\text{$\mathsf{\#P}$}\xspace}

\def\bpp{\text{$\mathsf{BPP}$}\xspace}
\def\p{\text{$\mathsf{P}$}\xspace}
\def\fp{\text{$\mathsf{FP}$}\xspace}

\newcommand{\eat}[1]{}

\def\win#1#2{\Win (#1\mathbin{;}#2)}
\def\lose#1#2{\text{Lose}(#1\mathbin{;}#2)}
\def\loseto#1#2#3{\text{L}(#1, #2\mathbin{;}#3)}

\newenvironment{repeatresult}[2]
{\vskip0.5em\par\textsc{#1} #2.\em}
{\vskip1em}

\title{Probabilistic Inference of Winners in Elections by Independent Random Voters} 

\author{Aviram Imber}{Technion, Israel}{}{}{}
\author{Benny Kimelfeld}{Technion, Israel}{}{}{}

\authorrunning{A.~Imber and B.~Kimelfeld}
\Copyright{Aviram Imber, Benny Kimelfeld}



\category{} 

\relatedversion{} 

\supplement{}

\nolinenumbers 

\begin{document}

\maketitle

\begin{abstract}
We investigate the problem of computing the probability of winning in an election where voter attendance is uncertain. More precisely, we study the setting where, in addition to a total ordering of the candidates, each voter is associated with a probability of attending the poll, and the attendances of different voters are probabilistically independent. We show that the probability of winning can be computed in polynomial time for the plurality and veto rules. However, it is computationally hard (\#P-hard) for various other rules, including $k$-approval and $k$-veto for $k>1$, Borda, Condorcet, and Maximin. For some of these rules, it is even hard to find a multiplicative approximation since it is already hard to determine whether this probability is nonzero. In contrast, we devise a fully polynomial-time randomized approximation scheme (FPRAS) for the complement probability, namely the probability of losing, for every positional scoring rule (with polynomial scores), as well as for the Condorcet rule.
\end{abstract}

\section{Introduction}
The theory of social choice targets the question of how voter preferences should be aggregated to arrive at a collective decision.  It spans centuries of research, from law making in Ancient Rome to the 20th century's fundamental theory and the recent study of the algorithmic and computing aspects---computational social choice.  (See~\cite{DBLP:reference/choice/0001CELP16} for an overview.) In the common formal setting, there are voters and candidates, each voter has an individual preference, namely ordering over the candidates, and a voting rule is applied to the preference profile (collection of orderings) in order to elect a winner. Past research has investigated various ways of incorporating situations of \e{uncertainty} in this setting, including incomplete and probabilistic preferences~\cite{konczak2005voting,DBLP:conf/aaai/BachrachBF10,DBLP:journals/ai/HazonAKW12,DBLP:conf/aaai/KenigK19, DBLP:conf/sigecom/FilmusO14}, probabilistic candidate validity~\cite{DBLP:conf/aaai/WojtasF12,DBLP:conf/aaai/BoutilierLOP14}, and probabilistic voter participation~\cite{10.2307/1953324,DBLP:conf/aaai/WojtasF12,DBLP:conf/aamas/WalshX12,DBLP:conf/atal/DeyB15,DBLP:conf/aaai/Micha020}.  Further settings include bribery with probabilistic voter cooperation~\cite{DBLP:conf/aaai/ChenXXGS19,DBLP:conf/ijcai/ChenXXGS19} and random elicitation for efficient outcome determination~\cite{DBLP:conf/sigecom/CaragiannisPS13,DBLP:conf/aldt/LuB11} .

In this work, we focus on the basic setting where voters are randomly drawn. In particular, the election outcome is probabilistic and, therefore, each candidate has a marginal probability of being elected.  More formally, each voter $v_i$ is associated with a total order over the candidates, and she casts a vote with probability $p_i$ (and steps outside the voter list with probability $1-p_i$).  Importantly, we assume that different voters are probabilistically independent. For a candidate $c$, the probability of winning is the probability that the random set of casting voters elects $c$. The aggregation of the preferences of these voters is done by a voting rule, and we consider several well-known alternatives: \e{positional scoring} rules, the \e{Condorcet} rule, and the \e{Maximin} rule (also known as the \e{Simpson} rule). We investigate the computational problem of calculating the probability that a given candidate $c$ is a winner.

{
\begin{table*}[t]
  \renewcommand{\arraystretch}{1.1}
  \newcommand{\thickline}{\Xhline{0.2ex}}
\centering
\caption{\label{tab:complexityOther} Overview of the complexity results.}
\scalebox{0.62}{
\begin{tabular}{l | c | c | c | c | c | c | c}
\thickline
  \textbf{Problem} & plurality, veto & $k$-\{approv., veto\} & Borda & $R(f, \ell)$ & other positional & Condorcet & Maximin \\
\thickline
\multirow{2}{*}{$\Pr[\text{win}]$} & 
\p & \sharpp-h & \sharpp-h & \sharpp-h for $(f, \ell) \neq (1,1)$  & \multirow{2}{*}{?} & \sharpp-h & \sharpp-h  \\
& [Thm.~\ref{thm:computeWinPluralityVeto}] & [Thm.~\ref{thm:computeWinApproval}] & [Thm.~\ref{thm:computeWinBorda}] & [Thm.~\ref{thm:computeWinThreeValues}] & & [Thm.~\ref{thm:computeWinCondorcet}] & [Thm.~\ref{thm:computeWinMaximin}] \\
\hline
\multirow{2}{*}{$\Pr[\text{win}]>0$?} & \p & \p & \np-c & \np-c for $(f, \ell) \neq (1,1)$  & \multirow{2}{*}{?} & \np-c & \np-c \\
& [Thm.~\ref{thm:computeWinPluralityVeto}] & [Thm.~\ref{thm:WinZeroBinaryScores}] & [Thm.~\ref{thm:winZeroBorda}] & [Cor.~\ref{cor:winZeroThreeValues}] & & [Thm.~\ref{thm:winZeroCondorcetMaximin}] & [Thm.~\ref{thm:winZeroCondorcetMaximin}] \\
\hline
Approximate & \p & FPRAS & FPRAS & FPRAS & FPRAS for poly. scores & FPRAS & \multirow{2}{*}{?} \\
$\Pr[\text{lose}]$ & [Thm.~\ref{thm:computeWinPluralityVeto}] & [Thm.~\ref{thm:losingApproxPolyScoreCondorcet}] & [Thm.~\ref{thm:losingApproxPolyScoreCondorcet}] & [Thm.~\ref{thm:losingApproxPolyScoreCondorcet}] & [Thm.~\ref{thm:losingApproxPolyScoreCondorcet}] & [Thm.~\ref{thm:losingApproxPolyScoreCondorcet}] & \\
\thickline
\end{tabular}
}
\end{table*}
}

Our contribution is summarized in Table~\ref{tab:complexityOther}. We first consider the complexity of exact evaluation (the first row of the table).  We start with positional scoring rules; recall that in such a rule, each voter assigns to each candidate a score based on the position of the candidate in the voter's ranking, and the vector of scores is shared among all voters.  For example, in $k$-approval the top $k$ positions get the score $1$ and the rest get $0$, and in $k$-veto the bottom $k$ positions get $0$ and the rest get $1$. The probability of winning can be computed in polynomial time for $k=1$ (i.e., the plurality and veto rules), but \#P-hard for every $k>1$. The problem is also hard for other positional scoring rules such as the Borda rule. Beyond the positional scoring rules, we also show that the problem is \#P-hard for the Condorcet and Maximin rules.

Since exact evaluation is often intractable, we consider the approximate version, that is, computing the probability of winning up to an approximation. There are various natural notions of approximation guarantees.  For an \e{additive} approximation, we can get a Fully Polynomial-Time Randomised Approximation Scheme (FPRAS) by straightforward sampling and averaging. So, we focus on a multiplicative approximation.  Yet, in some of the rules we consider it is intractable to get any approximation guarantee of a sub-exponential ratio, since it is already hard to decide whether the probability of winning is nonzero (as shown in the second row of Table~\ref{tab:complexityOther}). Exceptional are $k$-approval and $k$-veto where this decision problem is solvable in polynomial time, but the existence of an efficient approximation remains an open problem. Nevertheless, we can also consider another kind of a multiplicative approximation, where we consider the complement probability of \e{losing}. (In the additive variant, approximating the probability is the same as approximating its complement, but it is not so in the case of a multiplicative approximation.)  Inspired by the work of Kenig and Kimelfeld~\cite{DBLP:conf/aaai/KenigK19} on probabilistic preferences, we can apply the Karp-Luby-Madras technique~\cite{DBLP:journals/jal/KarpLM89} to establish a multiplicative FPRAS for the probability of losing.  This applies to \e{every} positional scoring rule (with the mild assumption that the scores are all polynomial in the number of voters and candidates) and the Condorcet rule; the problem remains open for Maximin.

While it is important to understand richer frameworks with more complex probabilistic modeling of voters, little has been known about the computational aspects of our basic setting, to the best of our knowledge. The closest studied problem is \e{lot-based voting}, where we uniformly select $k$ voters for a pre-determined number $k$~\cite{DBLP:conf/aamas/WalshX12}, and its generalization that entails a preprocessing step of choosing $k$ itself randomly from a given distribution~\cite{DBLP:conf/aaai/WojtasF12}.  Lot-based polling relates to, and can adopt complexity results on, the problem of \e{control} in elections, where the goal is to detect a small subset of voters of whom elimination can lead to a desired outcome~\cite{DBLP:conf/aaai/WojtasF12}. Yet, it is not clear how to draw conclusions from lot-based voting to out setting. The positive results in this model are based on counting, so the algorithms do not immediately apply to our setting where voters can have individual probabilities. As for the hardness results, in some cases we adapt proofs from that work to our setting, while in other cases the proofs are heavily based on the assumption that the cardinality of the voter set is given. Moreover, that work focused on the \e{exact} (but not \e{approximate}) computation of probabilities where our problem is \#P-hard in all the rules we consider with the exception of plurality and veto. Also related is the problem of predicting the winner of an election given sampled votes~\cite{DBLP:conf/atal/DeyB15,DBLP:conf/aaai/Micha020}. This setting is quite different: the voter set is deterministic, randomness is due to voter sampling, and the goal is to accurately predict the outcome.

The remainder of the paper is organized as follows. In Section~\ref{sec:preliminaries} we give preliminary definitions and terminology.  We discuss exact probability computation in Section~\ref{sec:exact}, the hardness of multiplicative approximation in Section~\ref{sec:approx-hardness}, and the FPRAS algorithms for the probability of losing in Section~\ref{sec:approx-losing}.
We conclude in Section~\ref{sec:conclusions}.
\section{Preliminaries}\label{sec:preliminaries}
We first give definitions, notation, and terminology that we use throughout the paper.

\subsection{Voting Profiles and Voting Rules}
We denote by $C = \set{c_1, \dots, c_m}$ the set of \emph{candidates} and $V = \set{v_1, \dots, v_n}$ the set of \emph{voters}. A \emph{voting profile} $\T = (T_1, \dots, T_n)$ consists of $n$ linear orders on $C$, where each $T_i$ is the ranking of (i.e., linear order over) $C$ by $v_i$. A \e{voting rule} is a function that maps every profile on $C$ to a set of \e{winners} from $C$.

A \e{positional scoring rule} $r$ is a series $\set{ \Vec{s}_m }_{m \in \natural^+}$ of $m$-dimensional score vectors $\Vec{s}_m = (\Vec{s}_m(1), \dots, \Vec{s}_m(m))$ of natural numbers where $\Vec{s}_m(1) \geq \dots \geq \Vec{s}_m(m)$ and $\Vec{s}_m(1) > \Vec{s}_m(m)$. We make the conventional assumption that $\Vec{s}_m(j)$ is computable in polynomial time in $m$.  Examples of positional scoring rules include the \emph{plurality} rule $(1, 0, \dots, 0)$, the \emph{$k$-approval} rule $(1, \dots, 1, 0, \dots, 0)$ that begins with $k$ ones followed by zeroes, the \emph{veto} rule $(1, \dots, 1, 0)$, the \emph{$k$-veto} rule $(1, \dots, 1, 0, \dots, 0)$ that starts with ones and ends with $k$ zeros, and the \emph{Borda} rule $(m-1, m-2, \dots, 0)$.

For two positive integers $f$ and $\ell$, we denote by $R(f, \ell)$ the three-valued rule with scoring vector $\vec{s}_m = (2, \dots, 2, 1, \dots, 1, 0, \dots, 0)$ that begins with $f$ occurrences of two and ends with $\ell$ zeros. For example, the scoring vector for $R(1, 1)$ is $\vec{s}_m = (2, 1, \dots, 1, 0)$.

Given a voting profile $\T = (T_1, \dots, T_n)$, the score $s(T_i, c, r)$ that the voter $v_i$ contributes to the candidate $c$ is $\Vec{s}_m(j)$ where $j$ is the position of $c$ in $T_i$. The score of $c$ in $\T$ is $s(\T, c, r) = \sum_{i=1}^n s(T_i, c, r)$ that we may denote simply by $s(\T, c)$ if $r$ is clear from context. A candidate $c$ is a \e{winner} (also referred to as \e{co-winner}) if $s(\T, c) \geq s(\T, c')$ for all candidates $c'$.\footnote{All of our results apply to, and can be easily adapted to, the \e{unique winner} semantics, where $c$ is a unique winner if $s(\T, c) > s(\T, c')$ for all candidates $c' \neq c$. }

For two candidates $c, c' \in C$, denote by $N_\T(c, c')$ the number of voters that prefer $c$ to $c'$.  A candidate $c$ is a \e{Condorcet winner} if $N_\T(c, c') > N_\T(c', c)$ for all $c' \neq c$. (Note that a Condorcet winner does not necessarily exist; when exists, it is unique.) Under the \e{Maximin} rule, the score of $c$ is $s(\T, c) = \min \set{N_\T(c, c') : c' \in C \setminus \set{c}}$ and a winner is a candidate with a maximal score.

\subsection{Probabilistic Voting Profiles}
In the setting we study, voters participate by being drawn randomly from $V$. We are given as part of the input a vector $(p_1, \dots, p_n) \in [0,1]^n$ of probabilities. Define a random variable $I \subseteq [n]$, where every $i \in [n]$ is in $I$ with probability $p_i$ and different indices are probabilistically independent. (Note our notation of $[n]=\set{1,\dots,n}$.) The random set of voters that participate in the election is $V' = \set{v_i : i \in I}$, and the resulting random profile is $\T' = \set{T_i}_{i \in I}$. The probability of $I$ being a subset $U \subseteq [n]$ is $\Pr[I = U] = \prod_{i \in U} p_i \prod_{i \in [n]\setminus U} (1-p_i)$.

Denote by $\win c {\T'}$ the event that $c$ is a winner of $\T'$ according to a positional scoring rule $r$. Similarly, denote by  $\lose c {\T'}$ the event that $c$ is not a winner. We will investigate the evaluation of the probabilities
$\Pr[\win c {\T'}]$ and $\Pr[\lose c {\T'}]$.

A \e{Fully-Polynomial Randomized Approximation Scheme} (FPRAS) for a probability function $p(x)$ is a randomized algorithm $A(x,\epsilon,\delta)$ that given $x$ for $p$ and $\epsilon,\delta\in(0,1)$, returns an $\epsilon$-approximation of $p(x)$ with probability at least $1-\delta$, in time polynomial in the size of $x$, $1/\epsilon$, and $\log(1/\delta)$. Formally, an FPRAS, satisfies:
\[\Pr[(1-\epsilon) p(x) \leq A(x, \epsilon, \delta) \leq (1+\epsilon) p(x)] \geq 1-\delta \,.\]
Note that this notion of FPRAS refers to a \e{multiplicative} approximation, and we adopt this notion implicitly, unless stated otherwise.

\subsection{Control by Voter Addition}
Let $r$ be a voting rule. In the problem of \e{Constructive Control by Adding Voters} (CCAV) we are given a set $C$ of candidates, a voting profile $\M$ of voters who are already registered, a voting profile $\Q$ of yet unregistered voters, a preferred candidate $c \in C$, and a bound $k \in \natural$. The goal is to test whether we can choose a sublist $\Q' \subseteq \Q$ of size at most $k$ such that $c$ is a winner of $\M \circ \Q'$, where $\M \circ \Q'$ is the concatenation of the two profiles $\M$ and $\Q'$. (The underlying rule $r$ will be clear from the context.)

\begin{theorem}[\cite{DBLP:conf/icaart/Lin11,mastersthesis/russell07,BARTHOLDI199227,DBLP:journals/jair/FaliszewskiHH11}]
The problem of CCAV is solvable in polynomial time for plurality, veto and $k$-approval for $k \leq 3$. CCAV is $\np$-complete for Borda, Condorcet, and $k$-approval for $k \geq 4$.
\end{theorem}

In the corresponding counting problem \#CCAV, studied by Wojtas and Faliszewski~\cite{DBLP:conf/aaai/WojtasF12}, the goal is to count the sublists $\Q' \subseteq \Q$ of size at most $k$ such that $c$ is a winner of $\M \circ \Q'$. They showed that \#CCAV is in \fp under plurality, and is \sharpp-complete under Condorcet, Maximin and $k$-approval for $k \geq 2$.

Note that in our model, where a random subset of voters is drawn, there is no restriction on the number of voters. For a voting rule $r$, we define the problem of \e{Constructive Control by Adding an Unlimited number of Voters} (CCAUV) as follows: We are given a set $C$ of candidates, a voting profile $\M$ of registered voters, a voting profile $\Q$ of yet unregistered voters, and a preferred candidate $c \in C$. We ask whether we can select a sublist $\Q' \subseteq \Q$ (of any cardinality) such that $c$ is a winner of $\M \circ \Q'$. We also denote the corresponding counting problem by \#CCAUV.

\subsection{Additional Notation}
For a set $A$ and a partition $A_1, \dots, A_t$ of $A$, we use $O(A_1, \dots, A_t)$ to denote an arbitrary linear order on $A$ that satisfies $a_1 \succ \dots \succ a_t$ for every $a_1 \in A_1, \dots, a_t \in A_t$.  A linear order $a_1 \succ \dots \succ a_t$ is also denoted as a vector $(a_1, \dots, a_t)$. The \e{concatenation} of two linear orders $(a_1, \dots, a_t) \circ (b_1, \dots, b_\ell)$ is $(a_1, \dots, a_t, b_1,\dots, b_\ell)$.
\section{ Exact Winning  Probability}\label{sec:exact}
In this section, we study the complexity of computing the probability
that a given candidate is a winner. Recall that the input consists of
a voting profile $\T$, a vector of voter probabilities and a candidate
$c$, and the goal is to compute $\Pr[\win c {\T'}]$.

\subsection{Tractability of Plurality and Veto}\label{sec:tractable-plurality-veto}
For plurality and veto, if the probabilities of all voters are identical, we can easily compute $\Pr[\win c {\T'}]$ in polynomial time by reducing the case of independent voters to the probability distribution studied by Wojtas and Faliszewski~\cite{DBLP:conf/aaai/WojtasF12}. For the general case, where the probabilities can differ, we need a different algorithm.

\begin{theorem}
\label{thm:computeWinPluralityVeto}
For the plurality and veto rules, $\Pr[\win c {\T'}]$ is computable in polynomial time.
\end{theorem}
\begin{proof}
We begin with plurality. Let $\T = (T_1, \dots, T_n)$ be a voting profile and let $(p_1, \dots, p_n)$ be the probabilities. In the plurality rule, every voter increases the score of a single candidate by 1. Since different voters are independent, the scores of different candidates are independent. Therefore, we get the following (due to the law of total probability).
\begin{align*}
    \Pr[\win c {\T'}] &= \sum_{s=0}^n \Pr[\win c {\T'}\land s(\T', c) = s] \\
    &= \sum_{s=0}^n \Pr[s(\T', c) = s] \cdot\prod_{c' \neq c} \Pr[s(\T', c') \leq s]
\end{align*}
We show that for all $c' \in C$ and $s \in \set{0, \dots, n}$, we can compute $\Pr[s(\T', c') = s]$ in polynomial time via dynamic programming. Then, by summing and multiplying these values, we can compute $\Pr[\win c {\T'}]$.

For a candidate $c'$, a number $t \leq n$ of voters and an integer score $0 \leq y \leq n$, define $N(t, y) \eqdef \Pr[\sum_{i \in I \cap [t]} s(T_i, c') = y]$. In particular, $N(n, s) = \Pr[s(\T', c') = s]$. For $t=0$ we have $N(t, y)=1$ if $y=0$ and $N(t, y)=0$ otherwise.

Let $t \geq 1$. If $c'$ is not ranked first in $T_t$, then the appearance of the voter $v_t$ does not affect the score of $c'$, hence $N(t, y) = N(t-1, y)$. If $c'$ is ranked first in $T_t$, then we consider two cases. If $t \in I$, then for the event $\sum_{i \in I \cap [t]} s(T_i, c') = y$ we need $\sum_{i \in I \cap [t-1]} s(T_i, c') = y-1$; otherwise, we need $\sum_{i \in I \cap [t-1]} s(T_i, c') = y$. Hence, we have the following.
\[ N(t, y) = p_t \cdot N(t-1, y-1) + (1-p_t) \cdot N(t-1, y) \]

For the veto rule, denote by $b(\T', c)$ the (random variable that
holds the) number of voters in $\T'$ who place $c$ at the bottom
position. Note that $b(\T', c)$ and $b(\T', c')$ are independent
for $c \neq c'$. We can write  the following.
\[ \Pr[\win c {\T'}] = \sum_{b=0}^n \Pr[b(\T', c) = b]\cdot \prod_{c' \neq c} \Pr[b(\T', c') \geq b]\]
For $c' \in C$ and $b \in \set{0, \dots, n}$, we can compute
$\Pr[b(\T', c') = b]$ in polynomial time similarly to $\Pr[s(\T', c')
= s]$
under plurality.
\end{proof}

\subsection{Connection to CCAUV}
In the remainder of this section, we show that computing
$\Pr[\win c {\T'}]$ is \sharpp-hard for several other voting rules.
For every voting rule we consider, we have a reduction from \#CCAUV to
computing $\Pr[\win c {\T'}]$, as follows. Let $C$, $\M$, $\Q$ and
$c \in C$ be an instance of \#CCAUV. Define probabilities for the
voters as follows. The voters of $\M$ participate with probability
$1$, and every voter of $\Q$ participates with probability $1/2$.
Denote by $k$ the number of voters of $\Q$ and by $\alpha(\Q, \M)$ the
number of subsets $\Q' \subseteq \Q$ such that $c$ is a winner of
$\M \circ \Q'$. We get that
$\Pr[\win c {\T'}] = 2^{-k} \alpha(\Q, \M)$. We conclude that \sharpp-completeness of \#CCAUV implies  \sharpp-hardness of
computing $\Pr[\win c {\T'}]$. Consequently, in this section the
proofs of our hardness results are essentially showing that \#CCAUV is
\sharpp-complete under the voting rule in consideration.

\subsection{$k$-Approval and $k$-Veto}
We begin with the hardness of $k$-approval for $k>1$ (where $k=1$ is the plurality rule
discussed in Section~\ref{sec:tractable-plurality-veto}).
\begin{theorem}
\label{thm:computeWinApproval}
For every fixed $k > 1$, computing $\Pr[\win c {\T'}]$ is \sharpp-hard under $k$-approval.
\end{theorem}
\begin{proof}
  We show a reduction from counting the (not necessarily perfect)
  matchings in a graph to \#CCAUV under $k$-approval. Given a graph $G = (U, E)$, we
  wish to compute the number of subsets $E' \subseteq E$ such that
  every vertex $u \in U$ is incident to at most one edge of $E'$. This
  problem is known to be
  \sharpp-complete~\cite{DBLP:journals/siamcomp/Valiant79}. Given a
  graph $G = (U, E)$ where $E = \set{e_1, \dots, e_m}$, define a set
  $C$ of candidates by $C = U \cup \set{c, d} \cup F$ where
  $F = F_0 \cup F_1 \cup \dots \cup F_m$ and
  $F_i = \set{f_{i, 1}, \dots, f_{i, k-2}}$. The voting profile is
  $\T = \M \circ \Q$ that we define next.

  The first part, $\M$, consists of a single voter
  $O(\set{c,d} \cup F_0, C \setminus (\set{c,d} \cup F_0))$.  Observe
  that the candidates of $\set{c,d} \cup F_0$ receive a score of 1
  from $\M$, and the other candidates receive 0. The second part is
  $\Q = (Q_1, \dots, Q_m)$. For $i \in [m]$, define
  $Q_i = O(e_i \cup F_i, C \setminus (e_i \cup F_i))$. Observe that
  the candidates of $e_i \cup F_i$ receive a score of 1 from $Q_i$,
  and the other candidates receive $0$. Also observe that here we are using the assumption that $k>1$, as  both endpoints of $e_i$ need to gain $1$ from $Q_i$.

  Let $\Q' \subseteq \Q$ and denote $\T' = \M \circ \Q'$. Since only the voter of $\M$ contributes to the score of $c$ and $d$, we have $s(\T', c) = s(\T', d) = 1$. Every $f \in F$ can get a positive score only from a single voter, hence $s(\T', f) \leq 1$. If $c$ is a winner, then $s(\T', u) \leq 1$ for all $u \in U$, and then $\set{e_i : Q_i \in \Q'}$ is a matching in $G$ because each voter $Q_i$ contributes a score of 1 to the vertices of $e_i$.  Conversely, if $E' \subseteq E$ is a matching in $G$, then $c$ is a winner of $\M \circ \set{Q_e}_{e \in E'}$.
  
Overall, the number of subsets $\Q' \subseteq \Q$ such that $c$ is a winner of $\M \circ \Q'$ is the number of matchings in $G$, as required. \end{proof}

Next, by a similar reduction to the proof of Theorem~\ref{thm:computeWinApproval}, we obtain hardness of $k$-veto for $k>1$ (where $k=1$ is the veto rule
discussed in Section~\ref{sec:tractable-plurality-veto}).
\def\thmcomputeWinVeto{For every fixed $k > 1$, computing $\Pr[\win c {\T'}]$ is \sharpp-hard under $k$-veto.}
\begin{theorem}
\label{thm:computeWinVeto}
\thmcomputeWinVeto
\end{theorem}
\begin{proof}
  We show a reduction from the problem of counting the edge covers in
  a graph to \#CCAUV under $k$-veto. Formally, given a graph $G = (U, E)$, the goal is to
  compute the number of subsets $E' \subseteq E$ such that every
  vertex $u \in U$ is incident to at least one edge of $E'$. This
  problem is known to be
  \sharpp-complete~\cite{DBLP:journals/rsa/BordewichDK08}.

  Given a graph $G = (U, E)$ where $E = \set{e_1, \dots, e_m}$, define
  a set $C$ of candidates by $C = U \cup \set{c, d} \cup F$ where
  $F = \set{f_1, \dots, f_{k-2}}$. The voting profile
  $\T = \M \circ \Q$ consists of $m+1$ voters. The first part, $\M$,
  consists of a single voter $O(U, \set{c, d} \cup F)$. Note that the candidates of $\set{c, d} \cup F$ receive a score of 0 from $\M$, and the other candidates receive 1. The second
  part is $\Q = (Q_1, \dots, Q_m)$ where
  $Q_i = O(C \setminus (e_i \cup F), e_i \cup F)$. Note that the
  candidates of $e_i \cup F$ receive a score of 0 from $Q_i$, and the
  other candidates receive 1.

  Let $\Q' \subseteq \Q$ and denote $\T' = \M \circ \Q'$. For every
  $c' \in C$, let $b(\T', c')$ be the number of voters in $\T'$ that
  rank $c'$ among the bottom $k$ positions. Note
  that $c$ is
  a winner in $\T'$ if and only if $b(\T', c) \leq b(\T', c')$ for all
  $c' \neq c$. We know that $b(\T', c) = b(\T', d) = 1$ and
  $b(\T', f) \geq 1$ for every $f \in F$. For $c$ to be a winner, we
  need $b(\T', u) \geq 1$ for all $u \in U$. Each voter $Q_i$ ranks
  the candidates of $e_i \cup F$ at the bottom $k$ positions, hence
  $b(\T', u) \geq 1$ for all $u \in U$ if and only if
  $\set{e : Q_e \in \Q_e}$ is an edge cover in $G$.

  Conversely, if $E' \subseteq E$ is an edge cover, then $c$ is a
  winner of $\M \circ \set{Q_e}_{e \in E'}$. Overall, the number of
  subsets $\Q' \subseteq \Q$ such that $c$ is a winner of
  $\M \circ \Q'$ is the number of edge covers in $G$.
\end{proof}

\subsection{Additional Positional Scoring Rules}
Next, we consider additional positional scoring rules: the Borda rule $(m-1, m-2, \dots, 0)$ and the ones of the form $R(f, \ell)$.  We use a technique of Dey and Misra~\cite{DBLP:conf/mfcs/DeyM17}.

\begin{lemma}[\cite{DBLP:conf/mfcs/DeyM17}]
\label{lemma:fixingScores}
Let $C = \set{c_1, \dots, c_m} \cup D$ be a set of candidates, where
$D$ is nonempty, and $\Vec{s}_{|C|}$ a normalized scoring
vector (i.e., the greatest common divisor of the scores is one). For every
$\Vec{x} = (x_1, \dots, x_m) \in \integer^m$, there exists
$\lambda \in \natural$ and a voting profile $\M$ such that
$s(\M, c_i) = \lambda + x_i$ for $i \in [m]$ and $s(\M, d) < \lambda$
for all $d \in D$. Moreover, the number of votes in $\M$ is polynomial
in $|C| \cdot \sum_{i=1}^m |x_i|$.
\end{lemma}

We first show hardness for the Borda rule, using Lemma~\ref{lemma:fixingScores}.

\begin{theorem}
\label{thm:computeWinBorda}
Computing $\Pr[\win c {\T'}]$ is \sharpp-hard under Borda.
\end{theorem}
\begin{proof}
  We show a reduction from counting the 
  matchings in a graph, as defined in the proof of
  Theorem~\ref{thm:computeWinApproval}, to \#CCAUV under
  Borda. Given a graph $G = (U, E)$ where $U = \set{u_1, \dots, u_n}$, define a set $C$ of candidates by
  $C = U \cup \set{c} \cup F$ where
  $F = \set{f_1, \dots, f_{n^3-n+1}}$. Note that $|C| = n^3 + 2$, hence the scoring vector is $(n^3+1, n^3, \dots, 0)$. The
  voting profile is $\T = \M \circ \Q$, as explained next.

  The first part, $\M$, is the profile that exists by
  Lemma~\ref{lemma:fixingScores} such that (for some constant
  $\lambda > 0$ from the lemma):
\begin{itemize}
    \item For every $f \in F$ we have $s(\M, f) < \lambda$.
    \item For every $u \in U$ we have $s(\M, u) = \lambda + n^5$.
    \item $s(\M, c) = \lambda + n^5 + 2n^3 - 1$.
\end{itemize}
The second part is $\Q = \set{Q_e}_{e \in E}$. For every edge
$e = \set{u,w} \in E$, define
$Q_e \eqdef O(e, F, U \setminus e, \set{c})$. From the construction we can see the following. First, the
scores of the candidates of $e$ satisfy $s(Q_e, u) \geq n^3$. Second,
for every other vertex $u' \in U \setminus e$ we have
$s(Q_e, u') \leq n-2$. Third, for every $f \in F$ we have
$s(Q_e, f) < n^3$. Finally, $s(Q_e, c) = 0$.

Let $E' \subseteq E$ be a set of edges.  For every $u \in U$, denote
by $\deg_{E'}(u)$ the number of edges of $E'$ incident to $u$. Define
$\Q' = \set{Q_e}_{e \in E'}$ and $\T' = \M \circ \Q'$. Since every
voter of $\Q'$ contributes a score of $0$ to $c$, we have
$s(\T', c) = s(\M, c) = \lambda + n^5 + 2n^3 - 1$. For every $f \in F$, we have
the following since $|E| \leq n^2$:
\begin{align*}
    s(\T', f) < \lambda + |E| \cdot n^3 \leq \lambda + n^5 < s(\T', c)
\end{align*}
For $u \in U$, if $\deg_{E'}(u) \leq 1$ then $u$ gains at most $n^3+1$
from edges that cover it, and at most $|E|(n-2)$ from the other edges of
$E'$. Overall,
\begin{align*}
    s(\T', u) &\leq n^3+1 + |E|(n-2) + s(\M, u) \leq n^3+1 + n^2(n-2) + \lambda + n^5 < s(\T', c)\,.
\end{align*}
Otherwise, if $\deg_{E'}(u) \geq 2$, then we have
\begin{align*}
    s(\T', u) \geq 2n^3 + s(\M, u) = \lambda + n^5 + 2n^3 > s(\T', c)\,.
\end{align*}
We can deduce that $c$ is a winner of $\T'$ if and only if
$\deg_{E'}(u) \leq 1$ for every $u \in U$, that is, $E'$ is a
matching. Since there is a
bijection between the subsets $E' \subseteq E$ and the sub-profiles
$\T' = \M \circ \Q'$, we get the correctness of the reduction.
\end{proof}

The next theorem states hardness for all positional scoring
rules of the form $R(f, \ell)$, except for the rule $R(1, 1)$ with the
scoring vector $(2, 1, \dots, 1, 0)$. The $R(1, 1)$ rule, which got a
considerable attention in the context of the possible-winner
problem~\cite{DBLP:journals/ipl/BaumeisterR12}, remains an open
problem. The proof
again uses Lemma~\ref{lemma:fixingScores}.

\def\thmcomputeWinThreeValues{ Computing $\Pr[\win c {\T'}]$ is \sharpp-hard under $R(f, \ell)$ whenever $(f, \ell) \neq (1, 1)$. }
\begin{theorem}
\label{thm:computeWinThreeValues}
\thmcomputeWinThreeValues
\end{theorem}
\begin{proof}
  First, consider the case where $f > 1$. We show a reduction from
  computing $\Pr[\win c {\T'}]$ under $f$-approval, which is
  \sharpp-hard by Theorem~\ref{thm:computeWinApproval}, to
  computing $\Pr[\win c {\T'}]$ under $R(f, \ell)$. Let
  $\T_1 = (T_1^1, \dots, T_n^1)$, be an instance for $f$-approval with
  probabilities $(p_1, \dots, p_n)$ over a set $C$ of
  candidates. Define a instance under $R(f, \ell)$ with candidate set
  $C' = C \cup D$ where $D = \set{d_1, \dots, d_\ell}$. The voters are
  $\T_2 = (T_1^2, \dots, T_n^2)$ where
  $T_i^2 = T_i^1 \circ (d_1, \dots, d_\ell)$ for every $i \in
  [n]$. The probabilities are $(p_1, \dots, p_n)$.

Observe that for every $i \in [n]$, the candidates of $D$ receive a score of $0$ from $T_i^2$, and for every $c \in C$ the score is $s(T_i^2, c) = s(T_i^1, c) + 1$. Since the probabilities are the same as in the instance under $f$-approval, we can deduce for every $c \in C$, the probability that $c$ is a winner of $\T_1'$ under $f$-approval is the same as the probability that $c$ is a winner of $\T_2'$ under $R(f,\ell)$.  

Next, assume that $\ell > 1$. For this case we show a reduction from
computing $\Pr[\win c {\T'}]$ under $\ell$-veto, which is \sharpp-hard
by Theorem~\ref{thm:computeWinVeto}, to computing $\Pr[\win c {\T'}]$
under $R(f, \ell)$. Let $\T_1 = (T_1^1, \dots, T_n^1)$ be an instance
for $\ell$-veto with probabilities $(p_1, \dots, p_n)$ over a set $C$
of candidates. Define a instance under $R(f, \ell)$ with the candidate set
$C' = C \cup D$ where $D = \set{d_1, \dots, d_f}$.

The voting profile $\T_2 = (T_1^2, \dots, T_n^2) \circ \M$ consists of
two parts. For the first part, for every $i \in [n]$ define
$T_i^2 = (d_1, \dots, d_f) \circ T_i^1$ and the probability is
$p_i$. The second part, $\M$, is obtained by applying
Lemma~\ref{lemma:fixingScores}, where (for some constant $\lambda > 0$
from the lemma) we have $s(\M, d) < \lambda$ for every $d \in D$ and
$s(\M, c) = \lambda + 3n$ for every $c \in C$. All voters of $\M$
appear with probability 1.

For every subset $A \subseteq [n]$, define a profile
$\T_2^A = \set{T_i^2}_{i \in A} \circ \M$. For every $d \in D$ the
score satisfies $s(\T_2^A, d) < 2|A| + \lambda \leq \lambda + 2n$ and
for every $c \in C$ we get $s(\T_2^A, c) \geq 0 + \lambda +
3n$. Therefore, every candidate in $C$ always defeats all candidates
in $D$. Furthermore, for every $i \in [n]$ and $c \in C$ we have
$s(T_i^2, c) = s(T_i^1, c)$. Overall, for every $c \in C$, the probability that $c$ is a winner of $\T_1'$ under $\ell$-veto equals the probability that $c$ is a winner of $\T_2'$ under $R(f,\ell)$. 
\end{proof}

\subsection{Condorcet and Maximin}
So far, we considered the complexity of computing the probability of
winning only for positional scoring rules. Next, we show hardness of
two rules of a different type: Condorcet and Maximin.  We begin with
the former.

\begin{theorem}
\label{thm:computeWinCondorcet}
Computing $\Pr[\win c {\T'}]$ is \sharpp-hard under Condorcet.
\end{theorem}
\begin{proof}
  We show a reduction from \#X3C to \#CCAUV under Condorcet. In \#X3C,
  we are given a vertex set $U = \set{u_1, \dots, u_{3q}}$ and a
  collection $E$ of 3-element subsets of $U$,
  and the goal is to count the $q$-element subsets of $E$ that cover
  $U$ using pairwise-disjoint sets. This problem is known to be
  \sharpp-complete~\cite{DBLP:journals/siamcomp/HuntMRS98}.

  Our reduction is an adaptation of the proof of Wojtas and
  Faliszewski~\cite{DBLP:conf/aaai/WojtasF12} that Condorcet-\#CCAV is
  \sharpp-complete. The reduction is as follows.  Let
  $U = \set{u_1, \dots, u_{3q}}$ and $E$ be an
  instance of \#X3C. The candidate set is $C = U \cup \set{c}$ and the
  voting profile is $\T_1 = \M_1 \circ \Q_1$.  The first part, $\M_1$,
  consists of $q-3$ voters with the preferences
  $(u_1, \dots, u_{3q}, c)$. The second part,
  $\Q_1 = \set{Q_1^e}_{e \in E}$, contains a voter for every set $e$
  in $E$, where $Q_1^e = O(e, \set{c}, U \setminus e)$. They showed a
  bijection between two collections: The sub-profiles $\Q_1' \subseteq \Q_1$ such that
    $|\Q'_1| \leq q$ and $c$ is a Condorcet winner of $\M_1 \circ \Q'_1$, and the subsets $E' \subseteq E$ that are exact covers.
  
We change the reduction as follows to show \sharpp-hardness of
\#CCAUV. First, we add another candidate $d$, so now
$C = U \cup \set{c, d}$. Second, the voting profile is
$\T_2 = \M_2 \circ \Q_2$. The first part, $\M_2$, consists of $q-1$
voters with the preferences $(u_1, \dots, u_{3q}, c, d)$, and two
voters with the preferences $(c, d, u_1, \dots, u_{3q})$.
We have the following for all $u \in U$.
\begin{align*}
  N_{\M_2}(u, c) &=  q-1 = N_{\M_1}(u, c) + 2\\
  N_{\M_2}(c, u) &=  2 = N_{\M_1}(c, u) + 2
\end{align*}
For $c,d$ we have $ N_{\M_2}(c, d) = q+1$ and $N_{\M_2}(d, c) = 0$. The second part, $\Q_2 = \set{Q_2^e}_{e \in E}$, contains the voter
$Q_2^e$ for each $e\in E$, where
$Q_2^e = O(e, \set{d}, \set{c}, U \setminus e)$.

To prove the correctness of our reduction, we will show a one-to-one
correspondence between our witnesses and those of Wojtas and
Faliszewski~\cite{DBLP:conf/aaai/WojtasF12}, that is, between the
sub-profiles $\Q_2' \subseteq \Q_2$ such that $c$ is a Condorcet
winner of $\M_2 \circ \Q'_2$, and the sub-profiles
$\Q_1' \subseteq \Q_1$ such that $|\Q'_1| \leq q$ and $c$ is a
Condorcet winner of $\M_1 \circ \Q'_1$.

Let $\Q_2' \subseteq \Q_2$, denote $\T_2' = \M_2 \circ \Q_2'$. Also define the sub-profiles
$\Q_1' = \set{Q_1^e : Q_2^e \in \Q_2'} \subseteq \Q_1$ and
$\T_1' = \M_1 \circ \Q_1'$. Observe that the following holds for all
$u \in U$:
\begin{align*}
    N_{\T_2'}(u, c) &= N_{\M_2}(u,c) + N_{\Q_2'}(u, c) = N_{\M_1}(u, c) + 2 + N_{\Q_1'}(u, c) \\
    &= N_{\T_1'}(u, c) + 2
\end{align*}
Similarly, we have $N_{\T_2'}(c, u) = N_{\T_1'}(c, u) + 2$,
$N_{\T_2'}(c, d) = q+1$ and $N_{\T_2'}(d, c) = |\Q_2'|$.
 Hence, for all $\Q_2' \subseteq \Q_2$, if $|\Q_2'| > q$ then $N_{\T_2'}(c, d) \leq N_{\T_2'}(d, c)$ and $c$ is not a Condorcet winner of $\T_2'$. Otherwise, $|\Q_2'| \leq q$. In this case we have $N_{\T_2'}(c, d) > N_{\T_2'}(d, c)$ and for every $u \in U$, $N_{\T_2'}(u, c) = N_{\T_1'}(u, c) + 2$ and $N_{\T_2'}(c, u) = N_{\T_1'}(c, u) + 2$.

Similarly, given a sub-profile $\Q_1' \subseteq \Q_1$ such that $|\Q_1'| \leq q$ and $c$ is a Condorcet winner of $\M_1 \circ \Q_1'$, we can define a sub-profile $\Q_2' = \set{Q_2^e : Q_1^e \in \Q_1'} \subseteq \Q_2$ and get that $c$ is a Condorcet winner of $\M_2 \circ \Q_2$. 

We conclude the claimed correspondence, and hence, we get a
polynomial-time reduction from \#X3C to  \#CCAUV.
\end{proof}

Next, we show hardness for the Maximin rule.
\begin{theorem}
\label{thm:computeWinMaximin}
Computing $\Pr[\win c {\T'}]$ is \sharpp-hard under Maximin.
\end{theorem}
\begin{proof}
  We show a reduction from \#X3C, as defined in the proof of
  Theorem~\ref{thm:computeWinCondorcet}, to \#CCAUV under Maximin. The
  reduction is an adaptation of the proof of Faliszewski,
  Hemaspaandra and
  Hemaspaandra~\cite{DBLP:journals/jair/FaliszewskiHH11} that CCAV is
  \np-complete under Maximin. Let $U = \set{u_1, \dots, u_{3q}}$ and
  $E$ be an instance of \#X3C. Define the candidate set
  $C = U \cup \set{c, d, w}$ and voting profile $\T = \M \circ \Q$, as
  detailed next.

  The $\M$ part consists of $4q$ voters as follows: $q$ voters with
  the preference
  $O(c, d, U, w)$, then $q-1$ voters with $O(c, U, w, d)$, then a
  single voter with $O(U, c, w, d)$, and lastly $2q$ voters with
  $O(d, w, U, c)$. The second part is $\Q = \set{Q_e}_{e \in E}$ where 
  $Q_e = O(w, U \setminus e, c, e, d)$.
  
  Let $\Q' \subseteq \Q$, define $\T' = \M \circ \Q'$ and $E' = \set{e \in E : Q_e \in \Q'}$. For $d$ we have that $N_\M(d, c) = 2q$ and, for all $u \in U$, that
  $N_\M(d, u) = 3q$ and $N_\M(d, w) = 3q$. Since the voters of $\Q'$
  rank $d$ at the bottom position, $\Q'$ does not affect
  the score of $d$ and $s(\T', d) = 2q$.

  For $w$ it holds that every voter of $\Q'$ ranks $w$ at the top
  position and, hence, we have that $N_{\T'}(w, c) = 2q + |\Q'|$, that
  $N_{\T'}(w, d) = q + |\Q'|$, and that $N_{\T'}(w, u) = 2q + |\Q'|$
  for all $u \in U$. Therefore $s(\T', w) = q + |\Q'|$. For $u \in U$ we have $N_{\T'}(u, d) = q + |\Q'|$, therefore, $s(\T', u) \leq q + |\Q'|$. 
  
  Finally, for $c$ we have $N_{\T'}(c, d) = 2q + |\Q'|$ and $N_{\T'}(c, w) = 2q$. For every $u \in U$, let $\deg_{E'}(u)$ be the number of sets of $E'$ incident to $u$, we get that $N_{\T'}(c, u) = 2q-1 + \deg_{E'}(u)$. We complete the proof by showing that the
  number of subsets $\Q' \subseteq \Q$ wherein $c$ is a winner of
  $\M \circ \Q'$ is the number of exact covers.

  First, suppose that $E' \subseteq E$ is an exact cover, that is,
  $|E'| = q$ and $\deg_{E'}(u) = 1$ for every $u \in U$. Define
  $\Q' = \set{Q_e : e \in E'}$ and $\T' = \M \circ \Q'$. From the
  above we have:
  \begin{itemize}
  \item $s(\T', d) = s(\T', w) = 2q$;
  \item $s(\T', u) \leq 2q$ for all $u\in U$;
  \item $s(\T', c) = 2q$ since $N_{\T'}(c, u) = 2q$ for all $u\in U$.
\end{itemize}
Hence, $c$ is a winner of $\T'$.

Conversely, let $\Q' \subseteq \Q$ be such that $c$ is a winner of
$\T' = \M \circ \Q'$, and let $E'$ be the corresponding subset of
$E$. If $|\Q'| > q$ then $s(\T', w) > 2q$ and
$s(\T', c) \leq N_{\T'}(c, w) = 2q$; hence, we get a contradiction. If
there exists $u \in U$ such that $\deg_{E'}(u) = 0$, then
$s(\T', c) \leq N_{\T'}(c, u) = 2q - 1$ and $s(\T', d) = 2q$, so we
again get a contradiction. We conclude that $|E'| \leq q$ and
$\deg_{E'}(u) \geq 1$ for all $u \in U$, therefore $E'$ is an exact
cover.
\end{proof}

To summarize the section, we established the row $\Pr[\text{win}]$ of
Table~\ref{tab:complexityOther} on the exact evaluation of the
probability of winning. In the next two sections, we discuss
approximate evaluation.

\section{Hardness of Approximation}\label{sec:approx-hardness}
Observe that there is an additive FPRAS for $\Pr[\win c \T]$ whenever we can test in polynomial time whether $c$ is a winner of a sampled (fully deterministic) profile. Such an FPRAS is obtained by taking the ratio of successes in trials wherein we sample voters (according to their distribution) and test whether $c$ is a winner; then, an FPRAS can be shown in standard ways (e.g., the Hoeffding's Inequality).

In this and the next section, we study the complexity of multiplicative approximation. Note that a multiplicative FPRAS implies an additive FPRAS, but not vice versa: if the probability is exponentially small, then $0$ is already an additive FPRAS, but not a multiplicative one. In probability estimation, it is often important to get a multiplicative approximation since, unlike the additive approximation, it allows for approximating ratios of probabilities (that are needed, e.g., for conditional probabilities) and for comparison between the likelihood of rare events. In the remainder of the paper, we restrict the discussion to multiplicative approximations, unless explicitly stated otherwise.

In this section, we argue that for most of the rules considered in the previous section, (multiplicative) approximation is also intractable and an FPRAS does not exist under conventional complexity assumptions.  We show it by proving that it is NP-hard to determine whether $\Pr[\win c {\T'}] = 0$ and, therefore, a multiplicative FPRAS for $\Pr[\win c \T]$ under $r$ implies that $\np \subseteq \bpp$.

As a general technique, observe that deciding whether $\Pr[\win c
{\T'}] > 0$ does not depend on the probabilities of voters (as long as this probability is nonzero). Hence, there is a polynomial-time reduction from the decision of $\Pr[\win c {\T'}] > 0$ to CCAUV and vice versa. In the first direction, each random voter with probability zero is ignored, each voter with probability one is put in $\M$, and each remaining voter is put in $\Q$. In the second direction, every voter of $\M$ participates with probability 1 and every voter of $\Q$ participates with probability 1/2 (or any other probability in $(0,1)$).

\subsection{Hardness Results}
We begin by showing hardness for the  Borda rule. 
\begin{theorem}
\label{thm:winZeroBorda}
For the Borda rule, CCAUV is $\np$-complete; hence, under Borda it is
$\np$-complete to decide whether $\Pr[\win c {\T'}] > 0$.
\end{theorem}
\begin{proof}
  We show a reduction to Borda-CCAUV from Borda-CCAV, which is known
  to be \np-complete~\cite{mastersthesis/russell07}. Let
  $\M_1, \Q_1, c^*$ and $k$ be an input for CCAV under Borda over a set $C_1$ of $m$ candidates, where
  $\Q_1 = (Q_1^1, \dots, Q_1^n)$. By the proof of
  Russell~\cite{mastersthesis/russell07} that CCAV is $\np$-complete
  for Borda, we can assume that all voters of $\Q_1$ rank $c^*$ at the
  top position. We construct an instance of CCAUV under Borda, where
  the candidate set is  $C_2 = C_1 \cup \set{d_1, d_2}$.

Let $\M_2$ be the profile of Lemma~\ref{lemma:fixingScores} such that (for some $\lambda > 0$):
\begin{itemize}
    \item $s(\M_2, d_1) < \lambda$ and $s(\M_2, d_2) = \lambda + 2mn + s(\M_1, c^*) - k$.
    \item For every $c \in C_1$ we have $s(\M_2, c) = \lambda + 2mn + s(\M_1, c)$.
\end{itemize}
The second profile $\Q_2 = (Q_2^1, \dots, Q_2^n)$ consists of $n$
voters, where $Q_2^i = (d_1, d_2) \circ Q_1^i$ for every
$i \in [n]$. 

Let $\Q_2'$ be a sublist of $\Q_2$, define $\Q_1' = \set{Q_1^i : Q_2^i \in \Q_2'}$, and let $\T_2' = \M_2 \circ \Q_2'$, $\T_1' = \M_1 \circ \Q_1'$. For $c \in C_1$ we have the following.
\begin{align}
    s(\T_2', c) &= s(\M_2, c) + s(\Q_2', c) = \lambda + 2mn + s(\M_1, c) + s(\Q_1', c) \notag \\
    &= \lambda + 2mn + s(\T_1', c) \label{eq:c}
\end{align}
Since we assume that all voters of $\Q_1$ rank $c^*$ at
the top position, for all $i \in [n]$ we have $s(Q_2^i, c^*) = s(Q_1^i, c^*) = m-1$, and overall
\begin{align*}
    s(\T_2', c^*) &= \lambda + 2mn + s(\M_1, c^*) + (m-1) |\Q_2'| 
\end{align*}
Finally, every voter of $\Q_2'$ contributes the scores $m+1$ and $m$ to $d_1$ and $d_2$, respectively, therefore
\begin{align}
    s(\T_2', d_1) &< \lambda + (m+1)|\Q_2'| \leq \lambda + (m+1)n \label{eq:d1} \\
    s(\T_2', d_2) &= \lambda + 2mn + s(\M_1, c^*) - k + m|\Q_2'| = s(\T_2', c^*) - k + |\Q_2'| \label{eq:d2}
\end{align}
From Equation~\eqref{eq:c} we can deduce that for every $c \in C_1 \setminus \set{c^*}$, $c^*$ defeats $c$ in $\T_1$ if and only if $c^*$ defeats $c$ in $\T_2$. From Equations~\eqref{eq:c} and~\eqref{eq:d1} we can deduce that $d_1$ is defeated by all candidates of $C_1$. From Equation~\eqref{eq:d2} we can deduce that $c^*$ defeats $d_2$ in $\T_2'$ if and only if $|\Q_2'| \leq k$. 

We show that there exists  $\Q_2' \subseteq \Q_2$ such that $c^*$ is a winner of $\M_2 \circ \Q_2'$ if and only if there exists $\Q_1' \subseteq \Q_1$ of size at most $k$ such that $c^*$ is a winner of $\M_1 \circ \Q_1'$. Let $\Q_2' \subseteq \Q_2$ such that $c^*$ is a winner of $\M_2 \circ \Q_2'$. In particular, $c^*$ defeats $d_2$, hence as we said $|\Q_2'| \leq k$. Let $\Q_1' = \set{Q_1^i : Q_2^i \in \Q_2'}$, we have $|\Q_1| \leq k$. For every  $c \in C_1 \setminus \set{c^*}$, $c^*$ defeats $c$ in $\M_2 \circ \Q_2'$, hence $c^*$ defeats $c$ in $\M_1 \circ \Q_1'$. Therefore $|\Q_1| \leq k$ and $c^*$ is a winner of $\M_1 \circ \Q_1'$.

Conversely, let $\Q_1' \subseteq \Q_1$ of size at most $k$ such that
$c^*$ is a winner of $\M_1 \circ \Q_1'$, define $\Q_2' = \set{Q_2^i : Q_1^i \in \Q_1'}$. By the same arguments as before, $c^*$ defeats every candidates of $C_1 \setminus \set{c^*}$ in $\M_2 \circ \Q_2'$, all candidates of $C_1$ defeat $d_1$, and $c^*$ defeats $d_2$ since $|\Q_2'| \leq k$. Hence, $c$ is a winner of $\M_2 \circ \Q_2'$. 
\end{proof}

The following two theorems show that CCAUV is $\np$-complete under
$R(f, \ell)$ whenever $(f, \ell) \neq (1,1)$.

\begin{theorem}
\label{thm:winZeroThreeValuesTop}
For every fixed $f \geq 2$ and $\ell \geq 1$, CCAUV is $\np$-complete under $R(f, \ell)$.
\end{theorem}
\begin{proof}
We show a reduction from the problem of \e{3-dimensional matching}
(3DM): Given three disjoint sets $X =
\set{x_1, \dots, x_q}$, $Y = \set{y_1, \dots, y_q}$ and $Z = \set{z_1,
  \dots, z_q}$ of the same size, and a set $E \subseteq X \times Y
\times Z$, is there a subset $E' \subseteq E$ consisting of $q$ pairwise-disjoint
triples? This problem is know to be \np-complete~\cite{DBLP:books/fm/GareyJ79}.
Given $X$, $Y$, $Z$ and $E = \set{e_1, \dots, e_m}$, we construct an
instance of CCAUV under $R(f, \ell)$. Denote $U = X \cup Y \cup
Z$. The candidate set is $C = U \cup W_1 \cup W_2 \cup \set{c, d}$
where $W_1 = \set{w_{1,1}, \dots, w_{1, f-2}}$ and
$W_2 = \set{w_{2,1}, \dots, w_{2, \ell-1}}$. Let $\M$ the profile
of Lemma~\ref{lemma:fixingScores} such that (for some
$\lambda > 0$):
\begin{itemize}
    \item $s(\M, c) = \lambda + 2m$ and $s(\M, d) < \lambda$;
    \item $s(\M, w) = \lambda$  for all $w \in W_1 \cup W_2$;
    \item $s(\M, x) = s(\M, y) = \lambda + 2m - 1$ for all $x \in X$
      and $y \in Y$;
    \item $s(\M, z) = \lambda + 2m + 1$ for all $z \in Z$.
\end{itemize}

The second profile $\Q = \set{Q_e}_{e \in E}$ consists of a voter for every triplet in $E$. For every $e = (x, y, z) \in E$ define
\[ Q_e = O(\set{x, y} \cup W_1, (U \setminus e) \cup \set{c, d}, \set{z} \cup W_2) \,. \]
Note that the candidates of $\set{x, y} \cup W_1$ receive a score of
$2$ from $Q_e$, the candidates of $\set{z} \cup W_2$ receive a score
of 0 from $Q_e$, and the remaining candidates receive 1.

We state some observations regarding the profile. Let
$\Q' \subseteq \Q$, define $E' = \set{e \in E : Q_e \in \Q'}$. Every voter of $\Q'$ contributes a score of 1 to $c$ and $d$,
hence their scores are:
\begin{align*}
 s(\M \circ \Q', c)& = \lambda + 2m + |\Q'| \\
s(\M \circ \Q', d) &< \lambda + |\Q'| \leq \lambda + m < s(\M \circ \Q', c)
\end{align*}
Similarly, for every $w \in W_1 \cup W_2$ the score satisfies $s(\M \circ \Q', w) \leq \lambda + 2m$, hence $c$ always defeats the candidates of $W_1 \cup W_2 \cup \set{d}$. 

For every $u \in U$, let $\deg_{E'}(u)$ be the number of triplet in $E'$ that are incident to $u$. For every $u \in X \cup Y$ we have
\begin{align}
    s(\M \circ \Q', u) &= \lambda + 2m-1 + |\Q'| + \deg_{E'}(u) = s(\M \circ \Q', c) - 1 + \deg_{E'}(u) \,.\label{eq:xy}
\end{align}
and for every $z \in Z$ we have
\begin{align}
    s(\M \circ \Q', z) &= \lambda + 2m+1 + |\Q'| - \deg_{E'}(z) = s(\M \circ \Q', c) + 1 - \deg_{E'}(z) \,.\label{eq:z}
\end{align}

We show that there exists $\Q' \subseteq \Q$ such that $c$ is a winner
of $\M \circ \Q'$ if and only if there is a 3DM. Let $E' \subseteq E$
be a 3DM, that is, $|E'| = q$ and $\deg_{E'}(u) = 1$ for all
$u \in U$. Define $\Q' = \set{Q_e : e \in E'}$. For every $u \in X \cup Y$, by Equation~\eqref{eq:xy}, we have $s(\M \circ \Q', u) = s(\M \circ \Q', c)$, and for every $z \in Z$, by Equation~\eqref{eq:z} we have $s(\M \circ \Q', z) = s(\M \circ \Q', c)$. Since  $c$ always defeats the candidates of $W_1 \cup W_2 \cup \set{d}$, we can deduce that $c$ is a winner of $\M \circ \Q'$.

Conversely, let $\Q' \subseteq \Q$ such that  $c$ is a winner of $\M \circ \Q'$, and let $E' = \set{e \in E : Q_e \in \Q'}$. For every $z \in Z$, since $c$ defeats every $z$, by Equation~\eqref{eq:z} we have $\deg_{E'}(z) \geq 1$. Every $z \in Z$ is covered by at least one set of $E'$, therefore $|E'| \geq q$. For every $u \in X \cup Y$, since $c$ defeats every $u$, by Equation~\eqref{eq:xy} we have $\deg_{E'}(u) \leq 1$. If $|E'| > q$ then there exists  $u \in X \cup Y$ for which  $\deg_{E'}(u) > 1$, hence a contradiction.

Overall, $|E'| = q$, each $z \in Z$ is covered by at least one set of
$E'$, and each $u \in X \cup Y$ is covered by at most one set of
$E'$. Therefore, $E'$ is a 3DM.
\end{proof}

By a reduction similar to the proof of Theorem~\ref{thm:winZeroThreeValuesTop}, we obtain the following.

\def\thmwinZeroThreeValuesBottom{ For every fixed $f \geq 1$ and $\ell \geq 2$, CCAUV is \np-complete under $R(f, \ell)$. }
\begin{theorem}
\label{thm:winZeroThreeValuesBottom}
\thmwinZeroThreeValuesBottom
\end{theorem}
\begin{proof}
  We again show a reduction from 3DM as defined in the proof of
  Theorem~\ref{thm:winZeroThreeValuesTop}. Given the input
  $X = \set{x_1, \dots, x_q}, Y = \set{y_1, \dots, y_q}, Z = \set{z_1,
    \dots, z_q}$ and $E = \set{e_1, \dots, e_m}$, we construct an
  instance of CCAUV under $R(f, \ell)$. Denote $U = X \cup Y \cup
  Z$. The candidate set is $C = U \cup W_1 \cup W_2 \cup \set{c, d}$
  where $W_1 = \set{w_{1,1}, \dots, w_{1, f-1}}$ and
  $W_2 = \set{w_{2,1}, \dots, w_{2, \ell-2}}$. 
  
The voting profiles $\M, \Q$ are constructed in the same way as in the proof of Theorem~\ref{thm:winZeroThreeValuesTop}, except that for every $e = \set{x, y, z} \in E$, the voter $Q_e$ contributes the score 2 to $x$ and 0 to $y,z$ (instead of contributing 2 to $x,y$ and 0 to $z$), and the scores of $\M$ are modified accordingly.
  
Formally, let $\M$ the profile of
  Lemma~\ref{lemma:fixingScores} such that (for some $\lambda > 0$):
\begin{itemize}
    \item $s(\M, c) = \lambda + 2m$ and $s(\M, d) < \lambda$;
    \item $s(\M, w) = \lambda$ for all $w \in W_1 \cup W_2$.
    \item $s(\M, x) = \lambda + 2m - 1$ for all $x \in X$, and
$s(\M, y) = s(\M, z) =   \lambda + 2m + 1$ for all $y\in Y$ and $z\in Z$.
\end{itemize}
The second profile is $\Q = \set{Q_e}_{e \in E}$. For every $e = \set{x, y, z} \in E$ define
\[ Q_e = O(\set{x} \cup W_1, (U \setminus e) \cup \set{c, d}, \set{y, z} \cup W_2)\,. \]

By the same arguments as in the proof of Theorem~\ref{thm:winZeroThreeValuesTop}, there exists $\Q' \subseteq \Q$ such that $c$ is a winner of $\M \circ \Q'$ if and only if there is a 3DM.
\end{proof}

From Theorems~\ref{thm:winZeroThreeValuesTop}
and~\ref{thm:winZeroThreeValuesBottom}
we conclude the following.
\begin{corollary}\label{cor:winZeroThreeValues}
  For all fixed $(f, \ell) \neq (1,1)$, it is
  $\np$-complete to determine whether $\Pr[\win c {\T'}] > 0$ under
  $R(f, \ell)$.
  \end{corollary}

For the rules Condorcet and Maximin, the proofs of
Theorem~\ref{thm:computeWinCondorcet} and
Theorem~\ref{thm:computeWinMaximin}, respectively, show a reduction
from \#X3C to computing the number of subsets $\Q' \subseteq \Q$ such
that $c$ is a winner of $\M \circ \Q'$. Since it is \np-hard to decide
whether there is an exact cover, it is also \np-hard to decide whether
there exists $\Q' \subseteq \Q$ such that $c$ is a winner of
$\M \circ \Q'$.

\begin{theorem}
  \label{thm:winZeroCondorcetMaximin}
  Under Condorcet and Maximin, CCAUV (and deciding whether $\Pr[\win c {\T'}] > 0$)
  is \np-complete.
\end{theorem}

\subsection{Tractable Zeroness for Binary Rules}
What about the positional scoring rules that are not covered by the
previous section? In particular, is there an FPRAS for $k$-approval
and $k$-veto for $k>1$ where an exact evaluation is \sharpp-hard
(Theorem~\ref{thm:computeWinApproval} and~\ref{thm:computeWinVeto})?
The question remains open. We do know, however, that the proof
technique of this section fails on them since it turns out that
zeroness can be decided in polynomial time for these rules (and, in
fact, every binary positional scoring rule).

\begin{theorem} \label{thm:WinZeroBinaryScores} For every binary
  positional scoring rule, CCAUV is solvable in polynomial
  time and, hence,  whether  $\Pr[\win c {\T'}] > 0$ can be decided in polynomial
  time.
\end{theorem}

\begin{proof}
  Let $(C, \M, \Q, c)$ be an instance of CCAUV. Let
  $\Q^* \subseteq \Q$ be the set of all voters of who contribute
  $1$ to $c$. We claim that if any $\Q' \subseteq \Q$
  is such that $c$ is a winner of $\M \circ \Q'$, then $c$ is a winner of
  $\M \circ \Q^*$. Hence, it suffices to test whether $c$ is a winner
  of $\M \circ \Q^*$.

  Let $\Q_1 \subseteq \Q$ such that $c$ is a winner of
  $\M \circ \Q_1$. For every voter of $\Q^* \setminus \Q_1$, if we add
  it to $\M \circ \Q_1$ then the score of $c$ increases by 1, and the
  score of every other candidate increases by at most 1. Hence, we can
  add all voters of $\Q^* \setminus \Q_1$ and $c$ remains a winner,
  that is, $c$ is a winner of $\M \circ \Q_2$ where
  $\Q_2 = \Q_1 \cup \Q^*$.  Now, for every voter in
  $\Q_2 \setminus \Q^*$, if we remove it from $\M \circ \Q_2$ then the
  score of $c$ is unchanged (since these voters contribute 0 to the
  score of $c$) and the score of the other candidates cannot
  increase. Therefore $c$ is a winner of $\M \circ
  \Q^*$. 
\end{proof}
\section{Approximate Losing Probability}\label{sec:approx-losing}
In the previous section, we showed that a multiplicative approximation
of $\Pr[\win c {\T'}]$ is often intractable since it is hard to
determine whether this probability is zero. In cases where it is
tractable to determine whether it is zero, the existence of an
approximation scheme (and even a constant-ratio approximation) remains
an open problem. In this section we show that, in contrast, we can
often get an FPRAS for the probability of \e{losing}:
$\Pr[\lose c {\T'}] = 1-\Pr[\win c {\T'}]$.

\begin{theorem} \label{thm:losingApproxPolyScoreCondorcet} There is a
  multiplicative FPRAS for estimating
  $\Pr[\lose c {\T'}] = 1-\Pr[\win c {\T'}]$ under every positional
  scoring rule with polynomial scores,  and under the Condorcet rule.
\end{theorem}

In the remainder of this section, we prove
Theorem~\ref{thm:losingApproxPolyScoreCondorcet}. We adapt the
technique of Karp-Luby-Madras~\cite{DBLP:journals/jal/KarpLM89} for
approximating the number of satisfying assignments of a DNF formula.

For a positional scoring rule, $c$ is not a winner if there exists
another candidate $d$ such that $s(\T', d) > s(\T', c)$. For a pair of
candidates $c \neq d$, let $\loseto c d {\T'}$ be the event that
$s(\T', c) > s(\T', d)$.  We can write
$\lose c {\T'} = \vee_{d \neq c} \loseto d c {\T'}$. Under Condorcet,
$c$ is not a winner if there exists $d \neq c$ such that
$N_{\T'}(d, c) \geq N_{\T'}(c, d)$. Therefore the event
$\loseto c d {\T'}$ is $N_{\T'}(c, d) \geq N_{\T'}(d, c)$ and we again
have
$\lose c {\T'} = \vee_{d \neq c} \loseto d c {\T'}$.

As before, let $I \subseteq [n]$ be the random variable that
represents the random set of voters who cast their vote. For the
Karp-Luby-Madras algorithm to estimate $\Pr[\vee_{d \neq c} \loseto d c {\T'}]$, we need
to perform the following tasks in polynomial time:
\begin{enumerate}
    \item Test whether $\loseto c d {\T'}$ is true in a given sample;
    \item Compute $\Pr[\loseto c d {\T'}]$ for every $c$ and $d$;
    \item Sample a random set $I$ of voters from the posterior
      distribution conditioned on the event $\loseto c d
      {\T'}$.
    \end{enumerate}
The first task is straightforward. For the other two tasks, we assume that we can compute $\Pr[\loseto c d {\T'} \mid I \cap S = S']$ in polynomial time, given a profile $\T$, a pair of candidates $c, d$ and voter sets $S \subseteq [n]$ and $S' \subseteq S$. Later, we will show how this can be done for particular rules. This solves the second task because $\Pr[\loseto c d {\T'}] = \Pr[\loseto c d {\T'} \mid I \cap \emptyset = \emptyset]$.

\begin{algorithm}[t]
\caption{Sampling from the posterior distribution conditioned on $\loseto c d {\T'}$.}\label{algorithm:sampleLoseTo}
\DontPrintSemicolon
Define $J_0 \eqdef \emptyset$\;
\For{$i = 1, \dots, n$} {
Compute $q_i \eqdef \Pr[i \in I \mid \loseto c d {\T'} \land I_{i-1} = J_{i-1}]$ \;
With probability $q_i$ define $J_i \eqdef J_{i-1} \cup \set{i}$, otherwise $J_i \eqdef J_{i-1}$\;
}
\Return $J \eqdef J_n$
\end{algorithm}

For the task of sampling from the posterior distribution conditioned
on $\loseto c d {\T'}$, for every $i \leq n$ we define a random
variable $I_i = I \cap \set{1, \dots, i}$.
Algorithm~\ref{algorithm:sampleLoseTo} presents the sampling
procedure, which constructs a random set of voters iteratively. It
begins with $J_0 \eqdef \emptyset$. Then, for every $i \in [n]$, at
the $i$th iteration, the algorithm defines
$J_i \eqdef J_{i-1} \cup \set{i}$ with probability
$\Pr[i \in I \mid \loseto c d {\T'}\land I_{i-1} = J_{i-1}]$, and
$J_i \eqdef J_{i-1}$ otherwise. The output is $J \eqdef
J_n$. 

We show that the algorithm is correct: for all $U \subseteq [n]$ we have
\begin{equation}\label{eq:eq-prob}
  \Pr[J = U] = \Pr[I = U \mid \loseto c d {\T'}]\,.
\end{equation} 
For that, we show by
induction on $i$ that for every $U \subseteq [i]$ we have
\begin{equation}\label{eq:eq-prob-induction}
  \Pr[J_i = U] = \Pr[I_i = U \mid \loseto c d {\T'}]\,.
\end{equation} 
Then, for
$i = n$ we get $\Pr[J = U] = \Pr[I = U \mid \loseto c d \T]$.

For the base $i=1$, the equality~\eqref{eq:eq-prob-induction} holds by the definition of the algorithm. Let $i > 1$, assume the equality holds
for $i-1$, and let $U \subseteq [i]$. If $i \in U$ then
\begin{align*}
    \Pr[J_i = U] = & \Pr[J_{i-1} = U \setminus \set{i}] \cdot \Pr[i \in J_i \mid J_{i-1} = U \setminus \set{i}] \\
  = & \Pr[I_{i-1} = U \setminus \set{i} \mid \loseto c d {\T'}] \cdot \Pr[i \in I \mid \loseto c d {\T'} \land I_{i-1} = U
    \setminus \set{i}] \\
  = & \Pr[I_i = U \mid \loseto c d {\T'}]\,.
\end{align*}
Otherwise, $i \notin U$, similarly,
\begin{align*}
    \Pr[J_i = U] &= \Pr[J_{i-1} = U \setminus \set{i}] \cdot \Pr[i \notin J_i \mid J_{i-1} = U \setminus \set{i}] = \Pr[I_i = U \mid \loseto c d {\T'}]\,.
\end{align*}
This concludes the correctness of the algorithm.

We now show that the algorithm can
be realized in polynomial time. In the the $i$th
iteration, we need to compute
$q_i = \Pr[i \in I \mid \loseto c d {\T'}\land I_{i-1} = J_{i-1}]$, that is:
\[ q_i = \frac{\Pr[\loseto c d {\T'} \land I_i = J_{i-1} \cup \set{i}]}{\Pr[\loseto c d {\T'}\land I_{i-1} = J_{i-1}]}\,. \]
For the denominator, observe that
\[ \Pr[\loseto c d {\T'} \land I_{i-1} = J_{i-1}] = \Pr[I_{i-1} = J_{i-1}] \times \Pr[\loseto c d {\T'} \mid I_{i-1} = J_{i-1}]\,. \]
We can compute $\Pr[I_{i-1} = J_{i-1}]$ directly from the definition:
\[ \Pr[I_{i-1} = J_{i-1}] = \prod_{j \in [i-1] \cap J_{i-1}} p_j \prod_{j \in [i-1] \setminus J_{i-1}} (1-p_j) \,. \]
For the second part, we have 
\[ \Pr[\loseto c d {\T'} \mid I_{i-1} = J{i-1}] = \Pr[\loseto c d {\T'} \mid I \cap [i-1] = J_{i-1}]\,, \]
which can be computed in polynomial time by our assumption that $\Pr[\loseto c d {\T'} \mid I \cap S = S']$ is computable in polynomial time. We can similarly
compute the numerator $\Pr[\loseto c d {\T'}, I_i = J_{i-1} \cup
\set{i}]$.

It remains to show that $\Pr[\loseto c d {\T'} \mid I \cap S = S']$
can indeed be calculated in polynomial time for $S \subseteq [n]$ and
$S' \subseteq S$. We start with positional scoring rules with
polynomial scores.

\begin{lemma}
\label{lemma:lostToGivenPolyScore}
For every positional scoring rule with polynomial scores, we can
compute $\Pr[\loseto c d {\T'} \mid I \cap S = S']$ in polynomial time
given a voting profile $\T$ over a set $C$ of candidates, a pair
$c \neq d$ of candidates, probabilities $(p_1, \dots, p_n$), and 
sets $S' \subseteq S$.
\end{lemma}
\begin{proof}
For every voter $v_i$ define $x_i = s(T_i, c) - s(T_i, d)$ and define
$D = \sum_{i \in S'} x_i$. Observe that
\def\bigmid{\mathrel{\bigr\vert}}
\begin{align*}
    &\Pr[\loseto c d {\T'} \mid I \cap S = S'] = \Pr \spar{\sum_{i \in I} x_i > 0 \bigmid I \cap S = S'} \\
    &= \Pr \spar{\sum_{i \in I \cap S} x_i  + \sum_{i \in I \setminus S} x_i > 0 \bigmid I \cap S = S' } = \Pr \spar{\sum_{i \in I \setminus S} x_i > -D}\,.
\end{align*}

We show that we can compute
$\Pr \spar{\sum_{i \in I \setminus S} x_i > -D}$ in polynomial
time. Assume, without loss of generality, that the set of voters
outside $S$ is $[n] \setminus S = \set{1, \dots, k}$.

Let $A$ be the sum of the negative values in $x_1, \dots, x_k$ and $B$ the sum of the positive values. For every $t \leq k$ and $y \in \integer$, define $N(t, y) = \Pr[\sum_{i \in I \cap [t]} x_i > y]$, our goal is to compute $N(k, -D)$. If $y < A$ then $N(t,y) = 1$ and if $y > B$ then $N(t,y) = 0$. Using dynamic programming, we compute $N(t, y)$ for all $t \leq k$ and $y \in \set{A, \dots, B}$ in polynomial time. Note that $B-A = \poly(n, m)$ since the scores are polynomial.

For $t=0$ we have $N(t, y) = 1$ if $y < 0$ and $N(t, y) = 0$ otherwise. Let $t \geq 1$, consider two cases. If $t \in I$ then for
the event $\sum_{i \in I \cap [t]} x_i > y$ we need $\sum_{i \in I
  \cap [t-1]} x_i > y-x_t$, otherwise we need $\sum_{i \in I \cap
  [t-1]} x_i > y$. Hence,
\[ N(t, y) = p_t \cdot N(t-1, y-x_t) + (1-p_t) \cdot N(t-1, y)\,. \]
This concludes the proof.
\end{proof}

Using a similar analysis to the proof of
Lemma~\ref{lemma:lostToGivenPolyScore}, we obtain the same result for the Condorcet rule.

\def\lemmalostToGivenCondorcet{ For the Condorcet rule, we can compute
  $\Pr[\loseto c d {\T'} \mid I \cap S = S']$ in polynomial time given
  a voting profile $\T$ over a set $C$ of candidates, a pair
  $c \neq d$ of candidates, probabilities $(p_1, \dots, p_n$) and sets
  $S' \subseteq S$. }
\begin{lemma}
\label{lemma:lostToGivenCondorcet}
\lemmalostToGivenCondorcet
\end{lemma}
\begin{proof}
Recall that for the Condorcet rule, the event $\loseto c d {\T'}$ is $N_{\T'}(c,d) \geq N_{\T'}(d,c)$. Let $V_c$ be the set of voters which rank $c$ higher than $d$, and let $V_d$ be the set of voters which rank $d$ higher than $c$. Define $D = |S' \cap V_c| - |S' \cap V_d|$, observe that
\begin{align*}
    &\Pr[\loseto c d {\T'} \mid I \cap S = S'] = \Pr \spar{|I \cap V_c| \geq |I \cap V_d| \mid I \cap S = S'} \\
    &= \Pr\big[ |(I \cap S) \cap V_c| + |(I \setminus S) \cap V_c|
      \geq |(I \cap S) \cap V_d| + |(I \setminus S) \cap V_d| \mid I \cap S = S'\big] \\
    &= \Pr \spar{ |(I \setminus S) \cap V_c| - |(I \setminus S) \cap V_d| \geq -D}\,.
\end{align*}

We show that we can compute this probability in polynomial time. Assume w.l.o.g. that the voters which are not in $S$ are $[n] \setminus S = \set{1, \dots, k}$. Let $A = -|[k] \cap V_d|$ and let $B = |[k] \cap V_c|$. For every $t \leq k$ and $y \in \integer$, define 
\[ N(t, y) = \Pr \spar{|(I \cap [t]) \cap V_c| - |(I \cap [t]) \cap V_d| \geq y} \,. \]
Our goal is to compute $N(k, -D)$. If $y < A$ then $N(t,y) = 1$ and if $y > B$ then $N(t,y) = 0$. Using dynamic programming, we compute $N(t, y)$ for all $t \leq k$ and $y \in \set{A, \dots, B}$ in polynomial time. Note that $B-A = O(n)$.

For $t=0$ we have $N(t, y) = 1$ if $y \leq 0$ and $N(t, y) = 0$ otherwise. Let $t \geq 1$, if $t \in V_c$ then 
\[ N(t,y) = p_t \cdot N(t-1, y-1) + (1-p_t) \cdot N(t-1, y)\,. \]
Otherwise, $t \in V_d$, then
\[ N(t,y) = p_t \cdot N(t-1, y+1) + (1-p_t) \cdot N(t-1, y)\,. \]
This concludes the proof.
\end{proof}

Note that Theorem~\ref{thm:losingApproxPolyScoreCondorcet} covers all
voting rules that we discuss in the paper, except for Maximin. The
correctness of Algorithm~\ref{algorithm:sampleLoseTo} holds for
Maximin, and in fact for every voting rule. However, to obtain an
FPRAS for estimating $\Pr[\lose c {\T'}]$ under Maximin, we need to be able to compute the probabilities $q_i$ of Algorithm~\ref{algorithm:sampleLoseTo} in polynomial time for Maximin. Whether this is the case remains an open problem.
\section{Conclusions}\label{sec:conclusions}
We studied the complexity of evaluating the marginal probability of wining in an election where voters are drawn independently at random, each with a given probability. The exact probability is computable in polynomial time for the plurality and veto rules, and \sharpp-hard for the other rules we considered, including classes of positional scoring rules, Condorcet and Maximin.  In some of these cases, it is also intractable to compute a multiplicative approximation of the probability of winning, since it is \np-complete to determine whether this probability is nonzero. For $k$-approval and $k$-veto, the zeroness problem is solvable in polynomial time, but the complexity of the multiplicative approximation remains unknown.  In contrast, we devised a multiplicative FPRAS for the probability of losing in all positional scoring rules with polynomial scores, in addition to Condorcet.

Several problems are left open for future investigation.  First, can we establish a full classification of the (pure) positional scoring rules in terms of their complexity for our problem, as in the case of the possible-winner problem over incomplete preferences~\cite{DBLP:journals/ipl/BaumeisterR12,DBLP:journals/jcss/BetzlerD10,DBLP:journals/jair/XiaC11}?  It might be the case the same classification holds, since our results so far are consistent with the possible winner (tractability for plurality and veto, and hardness for the rest).  To begin with, a specific rule of which absence stands out in Table~\ref{tab:complexityOther} is $(2,1,\dots,1,0)$, namely $R(1,1)$, that received much attention in the context of the possible winners~\cite{DBLP:journals/ipl/BaumeisterR12}.  Moreover, there are other voting rules to consider, such as Copeland and Bucking.  Second, while all of our reductions used only the probabilities $1/2$ and $1$, would it suffice to use just $1/2$ (i.e., counting the voter subsets where the candidate wins)? Finally, the existence of an FPRAS for Maximin remains unknown.

\bibliography{References}

\end{document}